\documentclass{article}

\usepackage{amsthm}
\usepackage{amsmath}
\usepackage{amssymb}
\usepackage{dlfltxbcodetips}
\usepackage{gastex}
\usepackage{hyperref}
\usepackage{natbib}
\usepackage{caption}
\usepackage{subcaption}
\usepackage{enumerate}
\usepackage{url}
\usepackage{arxiv}
\usepackage{graphicx}
\usepackage{color}
\usepackage{multirow}
\usepackage{placeins}
\usepackage{longtable}

\newtheorem{theorem}{Theorem}
\newtheorem{lemma}[theorem]{Lemma}
\newtheorem{proposition}[theorem]{Proposition}
\theoremstyle{definition}
\newtheorem{definition}[theorem]{Definition}
\newtheorem{remark}[theorem]{Remark}
\newtheorem{algorithm}{Algorithm}
\newtheorem{assumption}{Assumption}

\newcommand{\rr}{\mathbb{R}}
\newcommand{\nn}{\mathbb{N}}
\newcommand{\rk}{\mathbb{R}^m}

\newcommand{\new}{^{\mathrm{new}}}

\begin{document}

\numberwithin{equation}{section}
\numberwithin{theorem}{section}

\title{Informative simultaneous confidence intervals \\for graphical test procedures}

%\author{W. Brannath, L. Kluge and M. Scharpenberg}

\author{ Werner~Brannath\thanks{Correspondence to: Werner Brannath,\ \ \href{mailto:brannath@uni-bremen.de}{brannath@uni-bremen.de}}\\ 
  Institute for Statistics and\\ 
  Competence Center for\\
  Clinical Trials Bremen\\
  University of Bremen\\
  Bremen, Germany \\  
   \And
	Liane Kluge\\
	Competence Center for\\
	Clinical Trials Bremen\\
  University of Bremen\\
  Bremen, Germany \\
	 \And
	Martin Scharpenberg\\
  Competence Center for\\
	Clinical Trials Bremen\\
  University of Bremen\\
  Bremen, Germany \\
 }

\date{\today}

\maketitle

\begin{abstract}
%\noindent 

Simultaneous confidence intervals (SCIs) that are compatible with a given closed test procedure are often non-informative. More precisely, for a one-sided null hypothesis, the bound of the SCI can stick to the border of the null hypothesis, irrespective of how far the point estimate deviates from the null hypothesis. This has been illustrated for the Bonferroni-Holm and fall-back procedures, for which alternative SCIs have been suggested, that are free of this deficiency. These informative SCIs are not fully compatible with the initial multiple test, but are close to it and hence provide similar power advantages. They provide a multiple hypothesis test with strong family-wise error rate control that can be used in replacement of the initial multiple test. The current paper extends previous work for informative SCIs to graphical test procedures. The information gained from the newly suggested SCIs is shown to be always increasing with increasing evidence against a null hypothesis. The new SCIs provide a compromise between information gain and the goal to reject as many hypotheses as possible. The SCIs are defined via a family of dual graphs and the projection method. A simple iterative algorithm for the computation of the intervals is provided. A simulation study illustrates the results for a complex graphical test procedure.
\end{abstract}

\section{Introduction}

\subsection{Informative simultaneous confidence intervals}

Numerous multiple comparison procedures are available for a~large spectrum of settings, where more than one confirmatory assertion is desired so that the familywise error rate is controlled in the strong sense (see, for example, \citealp{HT87} or \citealp{D14} for an overview). Because of the increasing importance in practice, a~(draft of the) EMA/CHMP Guideline on multiplicity issues in clinical trials (\citealp{EMA17}) has been published. Besides error control, it emphasizes the importance of obtaining clinically interpretable results by providing confidence intervals that ``allow for consistent decision making with the primary hypothesis testing strategy". The implementation of such a~decision strategy is impeded by the fact that simultaneous confidence intervals (SCIs), which are both compatible with a~given multiple test, and informative, are not easy to find. Compatibility means that a~null hypothesis is rejected if and only if it is excluded from the confidence interval. 

We will call a confidence interval informative, if the information provided by the interval always increases (and only stays constant in the case that all gatekeepers for a hypothesis were not rejected) with increasing evidence against 
the corresponding null hypothesis. As a consequence, when a null hypothesis is rejected, the confidence interval will
have 
% be in the interior of the alternative hypothesis with 
a non-zero distance to the null hypotheses, except for the singular and usually negligible event that the corresponding (un-adjusted) p-value is equal to its final local level. Hence, informative confidence intervals almost always provide additional information to the mere hypothesis test.

To be more formal, let us consider $m$~null hypotheses $H_j:\theta_j \le 0$ for $j=1,\ldots,m$. We are interested in rejecting as many hypotheses as possible and at the same time obtain informative SCIs. By the possibilities of shifting, inverting and intersecting one-sided intervals, we restrict ourselves to left-sided hypotheses, without loss of generality. The SCIs are then $m$-dimensional rectangles that are bounded from the left, i.e., $\mathrm{SCI}=(L_1,\infty)\times\cdots\times(L_m,\infty)$. They are compatible with the rejection decisions if for $i=1,\ldots,m$, the null hypothesis~$H_i$ is rejected if and only if $L_i\ge 0$. We now give a~formal definition for an informative SCI.

\begin{definition} \label{def_informative}
We call a simultaneous confidence interval given by lower bounds $L:=(L_1,\ldots,L_m)$ \textit{informative} if
\begin{enumerate}[(a)]
\item $L_i > -\infty$ whenever $H_i$ has no gatekeeper or for at least one gatekeeper $H_j$ for $H_i$ we have $L_j>0$;
\item $L_i(X')>L_i(X)$ for two data sets $X'$ and $X$, whenever $X'$ provides more evidence against $H_i$ than $X$, $L_i(X)>-\infty$ and for all $j\not= i$ the evidence against $H_j$ is stronger in $X'$ or the same in both data sets. In case $L_i(X)=-\infty$ (i.e.\ all gatekeepers of $H_i$ are not rejected by the test induced by $L$) $L_i(X')\geq L_i(X)$  
%$L_j > 0$ for all hypotheses $H_j$ that are rejected.
\end{enumerate}
\end{definition}

\begin{remark}\label{rem_gatekeeping}\ \\[-1.5em] 
\begin{enumerate}[(i)]
\item  The restriction in the first item of Definition~\ref{def_informative} corresponds to the case where rejection of a~hypothesis~$H_i$ is only of interest when another or several other hypotheses (the ``gatekeepers'') are rejected. If the gatekeeper(s) cannot be rejected, then $H_i$ and its parameter $\theta_i$ are not considered at all and $L_i = -\infty$ is unavoidable.
\item We have intentionally left point (b) in Definition~\ref{def_informative} somewhat informal, namely with regard to meaning of the statement ``increasing evidence against $H_i$''. The mathematical definition of this must depend on the statistical model and hypothesis under investigation. When the test statistic and p-value is based on an estimate $\hat{\theta}_i$ of the parameter $\theta_i$ in $H_i:\theta_i\le 0$, then the evidence against $H_i$ increases with increasing 
estimate $\hat{\theta}_i$. 

\item We will see (and explicitly state) below that for our method and theory to apply, we need the existence of p-values $p_i(\mu_i)$ for all  shifted null hypotheses $H^{\mu_i}_i:\theta_i\le\mu_i$, which are all strictly decreasing with increasing evidence against $H_i$. This is usually the case when all these p-values are based on the same estimate $\hat{\theta}_i$ for $\theta_i$. We will present a typical example in Section 2.

\item %Point (b) in Definition~\ref{def_informative} implies that $L_i=0$ only in the case where $\hat{\theta}_i$
%$H_i:\theta_i\le 0$ is rejected with a p-value $p_i$ that is exactly equal to its (final) local level. 
If (b) of Defintion~\ref{def_informative} is satisfied for increasing $\hat{\theta}_i$ and the conditional distribution of the estimate $\hat{\theta}_i$ given the other estimates $\hat\theta_j$, $i\neq j$ is continuous, i.e.\ has a conditional Lebesgue density, then $L_i=0$ will occur only with probability zero and we get  $L_i>0$ whenever $H_i$ is rejected. This is the case, for instance, when the $m$ estimates $\hat{\theta}_1,\ldots,\hat{\theta}_m$ are multivariate normal.  

\item The property, that $L_i>0$ whenever $H_i$ is rejected, has been used in \citealp{BS14} and \citealp{SB14,SB15} as the defining feature of an informative SCI. It has been illustrated by simulations with multivariate normal estimates and formally verified in \citealp{SB14} under the assumption of continuously distributed p-values. As noted in \citealp{SB14}, with non-continuous estimates and p-values, even classical confidence bounds do not meet this definition, i.e.\ can be equal to the border of the null hypothesis with positive probability. Definition~\ref{def_informative} has the advantage to be independent of the (conditional) distribution of the estimates and p-values and to apply (as far as we can see) to all classical confidence bounds including those of single step simultaneous confidence intervals. 
\end{enumerate}
\end{remark}
%can easily be extended to compatible and informative SCIs, but this is at the cost of a~possibly smaller number of rejected hypotheses. 
Stepwise procedures reject more hypotheses on average than single-step tests 
and are therefore preferred in practice. \citealp{SB08} and \citealp{G08} have proposed SCIs for a~large class of stepwise procedures based on the closed testing principle. These SCIs are not always informative, e.g., if not all hypotheses are rejected then the confidence intervals of rejected hypotheses equal the whole alternative hypotheses, irrespective of how much the point estimate points into the alternative hypothesis. Hence, they contradict Definition~\ref{def_informative}  and they are only of limited use, because they do not provide any more information than the rejection itself. The recommendation of the (draft of the) EMA Guideline on multiplicity issues in clinical trials concerning this conflict of interest is the following: ``it is advised to use simple but conservative confidence interval methods, such as Bonferroni-corrected intervals". This is comprehensive with regard to the wish of having intervals that do not lead to misinterpretation, which is of greatest importance in practice. On the other hand, there is the need for a~compromise, because the recommended intervals are either compatible with the -- conservative -- hypothesis test, or informative but not compatible with a~more complex test. Since SCIs contain always more information than pure rejection decisions, a~natural way out of this conflict is to construct the multiple testing procedure by directly defining informative simultaneous confidence bounds~$L_1,\ldots,L_m$. These bounds are naturally consistent with the multiple test, which, by definition, rejects a~null hypothesis~$H_i$ if and only if $L_i\ge 0$.

Based on this idea, \citealp{BS14} have constructed informative simultaneous confidence intervals, which are always informative and uniformly more powerful than the Bonferroni procedure with regard to the number of rejected hypotheses. They can be seen as a~compromise between the powerful Bonferroni-Holm procedure and the informative Bonferroni procedure. Similar procedures were proposed for the hierarchical and the fallback test by \citealp{SB14,SB15}. All these procedures belong to the class of graphical test procedures introduced in \citealp{B09} which permit to account for preferences and hierarchies among the different null hypotheses $H_i$, $i=1,\ldots,m$. In this paper, we extend the previous idea to SCIs that are based on a given graphical test as defined in \citealp{B09}. We will see that the rejections made with the original graph cannot be exactly reproduced by our procedure. In contrast, the proposed confidence bounds will always be informative for all hypotheses which do not have gatekeepers, which, to our knowledge, is not possible for the original procedure. Actually, we can find a~trade-off between the number of rejections done by the original graph and the expected size of the confidence bounds by choosing the involved information weight $q$ defined later accordingly. The new confidence intervals are defined via a continuous family of graphical tests that are derived from the original one. They can numerically be calculated by an extension of the algorithm in \citealp{B09}. 

In the following subsection, we will recall the definition of the graphical procedures from \citealp{B09}, which sets the basic notation for the new approach. In Section~\ref{sec_method}, we introduce the new SCIs, first by a~formal definition, which motivates the approach, and subsequently by an iterative algorithm, which can be implemented numerically. We will see that both perspectives yield the same SCIs. In Section~\ref{sec_examples}, we link the extended setting to our previous SCIs proposed for the hierarchical, fallback and Bonferroni-Holm procedure (Brannath and Schmidt, 2014; Schmidt and Brannath, 2014, 2015), \nocite{SB14,SB15,BS14} 
and we give an example of a~more complex graph where the advantages of the new approach are demonstrated. We conclude with a~discussion in Section~\ref{sec_discussion}.
% Simulation study?

\subsection{Graphical multiple test procedures} \label{sec_B09}

A multiple testing procedure in \citealp{B09} is given by a~graph~$G^0$ which consists of
\begin{itemize}
\item initial local levels $\alpha_1,\ldots,\alpha_m\ge 0$ for the $m$ hypotheses such that $\sum_{j=1}^m \alpha_j = \alpha$, where $\alpha$ is the predefined significance level,
\item a transition matrix $(g_{ij})_{i,j=1,\ldots,m}$, where $g_{ij}$ is the weight by which the level of the $i$th hypothesis is shifted after its rejection to the $j$th hypothesis.
\end{itemize}
An example is given in Figure~\ref{fig_graph}. This graph is given in \citealp{B09} and is an example for a step-down test without order constraint from \citealp{B01}. Three treatments are compared with respect to efficacy and safety so that altogether six hypotheses are tested. The significance level is equally split across the three efficacy hypotheses, i.e., $\alpha_{E_1} = \alpha_{E_2} = \alpha_{E_3} = \alpha/3$ and $\alpha_{S_1} = \alpha_{S_2} = \alpha_{S_3} = 0$. Safety of a~treatment is tested only after the respective efficacy assertion has been shown. Hence, for all $i=1,2,3$, $H_{E_i}$ is a gatekeeper for $H_{S_i}$. There is no hierarchy between the three treatments. The corresponding transition matrix is
$$
\left(
\begin{matrix}
									   0 & 0 & 0 & 1 & 0 & 0 \\
                     0 & 0 & 0 & 0 & 1 & 0 \\
                     0 & 0 & 0 & 0 & 0 & 1 \\ 
                     0 & 0.5 & 0.5 & 0 & 0 & 0 \\ 
                     0.5 & 0 & 0.5 & 0 & 0 & 0 \\ 
                     0.5 & 0.5 & 0 & 0 & 0 & 0
\end{matrix}
\right),
$$
whereby the first three components belong to the efficacy hypotheses and the last three to safety.
%\begin{figure}
%\vspace{-0.3cm}
%      \includegraphics[scale=0.5]{beispiel_bretzmaurer.eps}
%      \caption{Example for a complex graphical procedure (Figure~8 from \citealp{B09}). $H_{E_j}$ is an efficacy hypothesis for the $j$th treatment and $H_{S_j}$ is a~safety hypothesis for the $j$th treatment.}
%  \label{fig_graph}
%\end{figure}

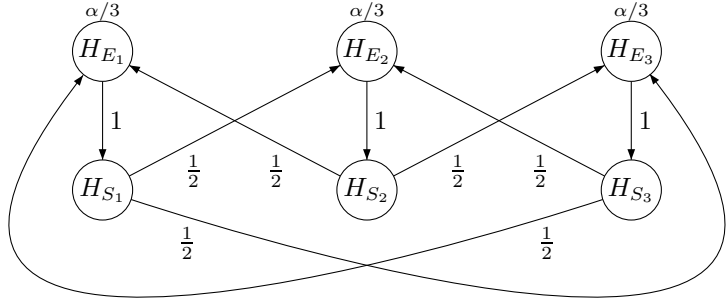
\begin{figure}
%\vspace{-0.3cm}
      \centering
      %\documentclass{article}

%\usepackage{amsthm}
%\usepackage{amsmath}
%\usepackage{amssymb}
%\usepackage{gastex}

%\begin{document}

\unitlength=2.5pt
\begin{picture}(100,42)(-5,-13)

 % \thinlines
	
	\node(A1)(0,21){$H_{E_1}$} 
  \node(A2)(40,21){$H_{E_2}$} 
	\node(A3)(80,21){$H_{E_3}$} 

	\node(B1)(0,0){$H_{S_1}$} 
  \node(B2)(40,0){$H_{S_2}$} 
  \node(B3)(80,0){$H_{S_3}$} 

  \put(0,27) {\makebox(0,0){\scriptsize$\alpha/3$}} 
  \put(40,27){\makebox(0,0){\scriptsize$\alpha/3$}} 
  \put(80,27) {\makebox(0,0){\scriptsize$\alpha/3$}} 
		
  \drawedge(A1,B1){$1$}  
  \drawedge(A2,B2){$1$} 
  \drawedge(A3,B3){$1$} 
  \drawedge[ELpos=30,ELside=r](B1,A2){$\frac{1}{2}$} 
  \drawedge[ELpos=30,ELside=l](B2,A1){$\frac{1}{2}$} 
  \drawedge[ELpos=30,ELside=r](B2,A3){$\frac{1}{2}$} 
  \drawedge[ELpos=30,ELside=l](B3,A2){$\frac{1}{2}$} 

  \drawbpedge[ELpos=20,ELside=r](B1,0,0,A3,-50,90){$\frac{1}{2}$}
  \drawbpedge[ELpos=20,ELside=l](B3,0,0,A1,-130,90){$\frac{1}{2}$}

	\end{picture}
	
	%\end{document}

%https://github.com/thomas-moulard/cours/blob/master/petri/gastex/gastex.sty

%-----------------------------------------------------------------------
% Edge between two nodes following a cubic Bezier curve.
% The first and last control points are the centers of the starting 
% and ending nodes resp, modified with the starting and ending offsets.
% The second and third control points are given by their polar coordinates
% relative to the first and last control points resp.
% 
%   \drawbpedge(startingNode,sa,sr,endingNode,ea,er){label}
%   \drawbpedge[parameter=value,...](startingNode,sa,sr,endingNode,ea,er){label}
%  
% Required arguments: 
%   startingNode : name of the starting node,
%   sa,sr : polar coordinates of the second control point relative to the
%           first control point (angle in degree and radius in \unitlength)
%   endingNode : name of the ending node,
%   ea,er : polar coordinates of the third control point relative to the
%           last control point (angle in degree and radius in \unitlength)
%   label : label of the edge. {} for no label.
% Optional argument:
%   [parameter=value,...]
      \caption{Example for a complex graphical procedure (Figure~8 from Bretz et al., 2009). $H_{E_j}$ is an efficacy hypothesis for the $j$th treatment and $H_{S_j}$ is a~safety hypothesis for the $j$th treatment.}
  \label{fig_graph}
\end{figure}
\nocite{B09}

The graphical algorithm is as follows: If a~hypothesis~$H_i$ can be rejected, its level is allocated to the other hypotheses according to the transition weights~$g_{ij}$, i.e.,
$$
\alpha_j\new = \alpha_j + \alpha_ig_{ij}, \quad j\ne i.
$$
 Arrows going from and to~$H_i$ are deleted and transition weights of the other arrows are updated as follows:
\begin{equation}
\label{update}
g_{jl}\new = \begin{cases}
							\frac{g_{jl} + g_{ji}g_{il}}{1-g_{ji}g_{ij}} & \text{if } j,l \ne i,\ j\ne l,\ g_{ji}g_{ij} \ne 1,\\
							0 & \text{else}.
						 \end{cases}
\end{equation}
In the next step, the remaining hypotheses are tested with their new local level. The graph is updated upon each rejection, until no more rejections are possible. It has been shown in \citealp{B09} that the set of rejected hypotheses is independent of the order in which these hypotheses are rejected. It was also proven that the procedure controls the familywise error rate at level $\alpha =  \sum_{j=1}^m \alpha_j$.

\section{The method} \label{sec_method}

\subsection{Definition of confidence bounds} \label{sec_definition}

Given observations (e.g. from a clinical trial), we assume that we have local p-values $p_j$, $j=1,\ldots,m$, for the null hypotheses $H_j = H_j^0:\theta_j\le 0$. Further, our observations allow us also to define p-values $p_j(\mu_j)$ for the shifted hypotheses $H_j^{\mu_j}:\theta_j\le\mu_j$. For example, if the estimate~$\hat{\theta}$ is (approximately) Gaussian, then $p_j = p_j(0) = 1-\Phi(\hat{\theta}_j/\mathrm{SE}_j)$ for some standard error~$\mathrm{SE}_j$. If the standard error ~$\mathrm{SE}_j$ is independent from $\theta$, then a~natural p-value for testing $H_j^{\mu_j}$ is 
\begin{equation}\label{eq_spval}
p_j(\mu_j) = 1-\Phi\left(\frac{\hat{\theta}_j-\mu_j}{\mathrm{SE}_j}\right).
\end{equation}
Usually, the p-values are strictly increasing and continuous in $\mu_j$, such that $\lim_{\mu_j\to -\infty}p_j(\mu_j)=0$ and $\lim_{\mu_j\to \infty}p_j(\mu_j)=1$. Moreover, as mentioned earlier, we will assume that for all $\mu_j\in \rr$ the p-values 
$p_j(\mu_j)$ strictly decrease with increasing evidence against $H_j$. This is obviously the case for the p-values in 
(\ref{eq_spval}), since the they decrease with increasing $\hat{\theta}_j$. Note that the mentioned properties apply also to the shifted p-values from $t$-distributed test statistics when the non-centrality parameter increase with increasing $\mu_j$ (which is usually the case).  
%All these will be assumed throughout the rest of this paper.

We start from a~graph as given in Section~\ref{sec_B09}, i.e., we have initial levels $\alpha_1,\ldots,\alpha_m\ge 0$, $\sum_{j=1}^m \alpha_j=\alpha$ and a~transition matrix $(g_{ij})_{i,j=1,\ldots,m}$. Our goal is to obtain simultaneous confidence intervals that: (i) reflect the structure of the testing procedure given by the graph, and (ii) are always informative (in the sense of Definition~\ref{def_informative}). Throughout this article, we will assume that the graph is complete in the following sense.
\begin{assumption} \label{ass_complete}
We assume that for all $i=1,\ldots,m$, the transition weights starting from $H_i$ sum up to one,
i.e.\ $\sum_{j=1}^m g_{ij} = 1$ (and that all $g_{ii}=0$).
\end{assumption}
\noindent We will explain in Remark~\ref{rem_extension} below how our procedure can be adapted if the graph is not complete.

\paragraph{Basic projection method.} By modifying the given graph, we construct weighted Bonferroni tests for the intersection of shifted hypotheses $H^{\mu}=H_1^{\mu_1}\cap\ldots\cap H_m^{\mu_m}$ for each $\mu=(\mu_1,\ldots,\mu_m)\in\rk$. We will reject the intersection hypothesis $H^{\mu}$ globally if and only if at least one of the hypotheses $H_j^{\mu_j}$ can be rejected at its local level $\alpha_j^{\mu}$ where $\sum_{j=1}^m \alpha_j^{\mu}=\alpha$ for all $\mu$. We can then define the $m$-dimensional confidence set 
$$C=\{\mu\in\rk:H^{\mu} \text{ is not rejected}\}=\{\mu\in\rk:\min_{\substack{j=1,\ldots,m\\ \alpha_j^{\mu}>0}} p_j(\mu_j)/\alpha^\mu_j>1\}.$$ 
By construction, the coverage probability of~$C$ is at least $(1-\alpha)$. 

We next construct the simultaneous confidence bounds~$L_j$, $j=1,\ldots,m$, by projection, i.e., the simultaneous confidence interval $(L_1,\infty)\times\cdots\times(L_m,\infty)$ is the smallest rectangle containing the open set~$C$. Since it contains~$C$, its coverage probability is not smaller than $1-\alpha$. 

\paragraph{Dual graphs and resulting weighted Bonferroni tests.} For each $\mu\in\rk$ we will construct local levels $\alpha_1^{\mu},\ldots,\alpha_m^{\mu}$ summing up to~$\alpha$. They will depend on a~parameter $q \in (0,1)$, which we will call \emph{information weight}. We explain the significance of~$q$ in Section~\ref{sec_properties}. 

For given $\mu\in\rk$ the local levels are constructed in two steps: in the first step we define a dual graph that contains all shift null hypotheses $H^{\mu_j}_j$ 
and some (not necessarily all) initial null hypotheses as nodes; in the second step we reject all initial null hypotheses in this graph to obtain the levels $\alpha_j^{\mu}$ for $H^{\mu_j}_j$, $j=1,\ldots, m$. 

\textit{First step.}
We define a new graph $G^{\mu}$ by modifying the given graph $G$ as follows: 
\begin{itemize}
\item for all $j$ with $\mu_j \le 0$:
\begin{itemize}
\item delete all paths starting at $H_j$, i.e. set $g_{ji}^{\mu}=0$ for all $i$;
\item replace $H_j$ by $H^{\mu_j}_j$.
\end{itemize}
\item for all $j$ with $\mu_j > 0$:
\begin{itemize}
\item add a node for the hypothesis $H_j^{\mu_j}$ to the graph with local level~$0$;
\item introduce an arrow from $H_j$ to $H_j^{\mu_j}$ with transition weight $q^{\mu_j}$;
\item change all transition weights starting from $H_j$ to $g_{ji}^{\mu}=g_{ji}(1-q^{\mu_j})$.
\end{itemize}
\end{itemize}

The resulting graph $G^{\mu}$ contains $m$ nodes with the starting levels $\alpha_1,\ldots,\alpha_m$ plus up to $m$ nodes with initial level zero. Note that all shifted hypotheses $H_j^{\mu_j}$ are contained in the graph. It is still a valid graph within the framework of \citealp{B09}, i.e., the row sums of the transition matrix are equal to~$1$. Figure~\ref{fig_graphs} shows an example of a graph $G$ and the resulting graph~$G^{\mu}$.

The rationale for the graph $G^{\mu}$ is that we can imagine it as the test for a~hypothesis~$H^\mu=H_1^{\mu_1}\cap\ldots\cap H_m^{\mu_m}$ that has not yet been rejected. If $\mu_j \le 0$ for some~$j$, then also the null hypothesis~$H_j=H^0_j$ has not yet been rejected and therefore, full level has to be given to~$H^{\mu_j}_j$ and no transfer to other hypotheses takes place. If however $\mu_j > 0$, then testing~$H_j^{\mu_j}$ increases the information on $\mu_j$ after the null hypothesis $H_j$ has been rejected. To this end, the graph contains $H_j$ and $H_j^{\mu_j}$ with a	non-zero transition weight (namely $q^{\mu_j}$) from $H_j$ to~$H_j^{\mu_j}$.

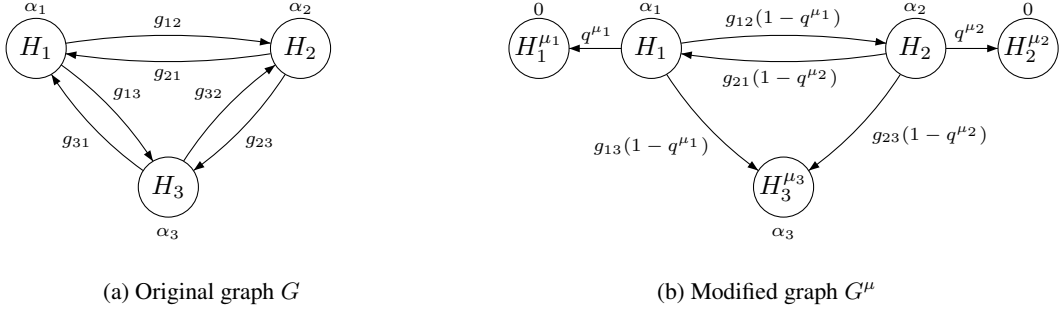
\begin{figure}
%\vspace{-0.3cm}
  \begin{subfigure}{0.4\textwidth}
      \centering
      \unitlength=2.5pt
\begin{picture}(60,37)(-5,-3)

  \thinlines
	\newcommand{\cb}{}
	
	%Original-Graph
	
	\node(A1)(0,21){$H_{1}$} 
  \node(A2)(40,21){$H_{2}$} 
  \node(A3)(20,0){$H_{3}$} 

  \put(0,27) {\makebox(0,0){\scriptsize $\alpha_1$}} 
  \put(40,27){\makebox(0,0){\scriptsize $\alpha_2$}} 
  \put(20,-7) {\makebox(0,0){\scriptsize $\alpha_3$}} 
		
  \gasset{curvedepth=2}

   \drawedge(A1,A2){\scriptsize $g_{12}$}  
   \drawedge(A2,A3){\scriptsize $g_{23}$} 
   \drawedge(A3,A2){\scriptsize $g_{32}$} 
   \drawedge(A1,A3){\scriptsize $g_{13}$} 
   \drawedge(A3,A1){\scriptsize $g_{31}$} 
   \drawedge(A2,A1){\scriptsize $g_{21}$} 
	
	\end{picture}
      \caption{Original graph~$G$}
  \end{subfigure}
  \begin{subfigure}{0.5\textwidth}
      \centering
      \unitlength=2.5pt
\begin{picture}(75,37)(-20,-3)

  \thinlines
	\newcommand{\cb}{}

	%Anfang: Original-Graph
	
	 \node(A1)(0,21){$H_{1}$} 
   \node(A2)(40,21){$H_{2}$} 
   \node(A3)(20,0){$H_{3}^{\mu_3}$} 

    \put(0,27) {\makebox(0,0){\scriptsize $\alpha_1$}} 
    \put(40,27){\makebox(0,0){\scriptsize $\alpha_2$}} 
    \put(20,-7) {\makebox(0,0){\scriptsize $\alpha_3$}}

	%Schritt 3: neue Beschriftung der Kanten

  \gasset{curvedepth=-2}

  \drawedge[ELpos=40,ELdistC=y,ELdist=-8](A1,A3){{\cb\scriptsize $g_{13}(1-q^{\mu_1})$} }
		
  \gasset{curvedepth=2}

  \drawedge(A1,A2){{\cb\scriptsize $g_{12}(1-q^{\mu_1})$}} 
  \drawedge[ELpos=30](A2,A3){{\cb\scriptsize $g_{23}(1-q^{\mu_2})$} }
  \drawedge(A2,A1){{\cb\scriptsize $g_{21}(1-q^{\mu_2})$} }
	
	%Schritt 2: KI-Hypothesen dazu mit q^mu Beschriftung
	
	\node(A1x)(-17,21){$H_{1}^{\mu_1}$} 
	\node(A2x)(57,21){$H_{2}^{\mu_2}$} 

  \put(-17,27) {\makebox(0,0){\scriptsize $0$}} 
  \put(57,27){\makebox(0,0){\scriptsize $0$}} 

 \gasset{curvedepth=0}
	
  \drawedge[ELdistC=y,ELdist=-2](A1,A1x){{\cb\scriptsize $\ q^{\mu_1}$}} 
  \drawedge(A2,A2x){{\cb\scriptsize $\ q^{\mu_2}$} }

\end{picture}
      \caption{Modified graph~$G^{\mu}$}
 \end{subfigure}
  \caption{From the original graph (a) to a graph (b) where the shifted hypotheses $H_1^{\mu_1}$ and $H_2^{\mu_2}$ with $\mu_1,\mu_2 > 0$ are added, while the null hypothesis $H_3$ is transformed to the shifted hypothesis~$H_3^{\mu_3}$ with $\mu_3\le 0$.}
  \label{fig_graphs}
\end{figure}

\textit{Second step.}
We reject in $G^{\mu}$ all null hypotheses $H_j=H^0_j$, performing the update algorithm of \citealp{B09}. The final graph contains only the shifted hypotheses $H^{\mu_j}_j$ with some new levels $\alpha_j^{\mu}$ and no arrows between these hypotheses. We will explain in the Appendix why the local levels $\alpha_j^{\mu}$ always satisfy
\begin{equation} \label{eq_alpha}
\sum_{j=1}^m\alpha_j^{\mu} = \alpha.
\end{equation}
Basically, this is because we started with a valid graph $G^{\mu}$ and performed only rejections according to \citealp{B09}, and because of Assumption~\ref{ass_complete}. However, some more technical issues have to be discussed for a~proof. To summarize, we have obtained a~level-$\alpha$ test for each $\mu\in\rk$, which rejects~$H^{\mu}$ if at least one of the local p-values satisfies $p_j(\mu_j) \le \alpha_j^{\mu}$.

According to the projection method described at the beginning of this section, the new lower simultaneous confidence bounds are then defined as follows.
\begin{definition}\label{def_iSCI}
Given the p-values $p_j(\mu_j)$, $\mu\in\rk$ and $j=1,\ldots,m$, the initial graph $G^{0}$, we define, based on the local levels $\alpha^\mu_j$ resulting from the above two steps with the dual graph $G^\mu$, 
%$\mu\in \rk$ and $j=1,\ldots,m$, 
the lower simultaneous bounds
\begin{equation}\label{eq_isci}
L_j=\max\{\mu_j:p_j(\mu_j)\le \alpha^{\mu'}_j \text{ for all }\mu'\in\rk\text{ with }\mu'_j=\mu_j\},\ j=1,\ldots,m,
\end{equation}
which have simultaneous coverage probability of at least $1-\alpha$ by construction.
\end{definition}

\subsection{Numerical algorithm to calculate the confidence bounds}

We now prove a~property of the local levels $\alpha_j^{\mu}$, which allows us to obtain the bounds~$L_j$ by a~simple numerical algorithm without the need to test the uncountable many intersection hypotheses~$H^{\mu}$ for $\mu\in\rk$ and calculate the maximum in (\ref{eq_isci}).

\begin{proposition} \label{prop_nu}
For all $\mu=(\mu_1,\ldots,\mu_m)\in\rk$, the local levels $\alpha_j^{\mu}$ derived in Section~\textup{\ref{sec_definition}} are of the form
$$
\alpha_j^{\mu} = q^{\mu_j\vee 0}\nu_j(\mu)\alpha,
$$
where $\nu_j\colon\rk\to\mathbb R_{\geq 0}$ is continuous and non-decreasing in each component.
\end{proposition}

\begin{proof}
We know or can see from \citealp{B09} and the up-date algorithm therein that the graphical test algorithm has the following properties:
\begin{enumerate}
\item The order in which hypotheses are rejected does not influence the local levels of the final graph.
\item After every rejection step, the new transition weights are continuous and non-decreasing functions of the old transition weights.
\item After every rejection step, the new level of a hypothesis and the new weights starting from this hypothesis are independent of the old levels of any other, non-rejected hypothesis and independent of all old transition weights that go from or to the other, non-rejected hypotheses.
\item The new level of a hypothesis is a~linear combination of~$\alpha_1,\ldots,\alpha_m$.
\end{enumerate}
Fix~$j$ and $\mu$. As already described, the level $\alpha_j^{\mu}$ for $H_j^{\mu_j}$ arises from rejecting all null hypotheses~$H_k=H^0_k$ in the graph $G^{\mu}$. If $\mu_j \le 0$, then only $H_k^{0}$ with $k\ne j$ are rejected. By the properties 2.~and 3., the final local level for $H_j^{\mu_j}$ is then a continuous and non-decreasing function of one or more $g_{kl}(1-q^{\mu_k})$ for $k\ne l$, and by 4., it is a~multiple of~$\alpha$. Hence, it is of the form $\nu_j(\mu)\alpha$ for a continuous and non-decreasing function~$\nu_j$.

If $\mu_j > 0$, then by 1., we can w.l.o.g.\ reject at first all $H_k$ with $k\ne j$ and finally reject $H_j$. This last step multiplies the level $\nu_j(\mu)\alpha$ by $q^{\mu_j}$ so that the desired representation is obtained.
\end{proof}

%\begin{remark}
%In the case of the adapted procedure mentioned in Remark~\ref{rem_extension} for non-complete graphs, the expression $q^{\mu_j\vee 0}$ in Proposition~\ref{prop_nu} would have to be replaced by some non-increasing function. For easier reading, we only consider the simpler expression~$q^{\mu_j\vee 0}$ in the following and also in the Appendix, although the arguments are identical.
%\end{remark}

We remind that $H^{\mu}=H_1^{\mu_1}\cap\ldots\cap H_m^{\mu_m}$ is rejected if, for some $j$, we have $p_j(\mu_j) \le \alpha_j^{\mu}$, i.e. by Proposition~\ref{prop_nu}, if
\begin{equation}
\label{eq_iteration}
p_j(\mu_j)q^{-(\mu_j\vee 0)} \le \nu_j(\mu)\alpha.
\end{equation}

Note that both, the left-hand and the right-hand side of~\eqref{eq_iteration}, are non-decreasing in~$\mu_j$. Based on this property, we can define an iterative algorithm for the calculation of the lower confidence bounds.

\begin{algorithm} \label{alg_ki}
Find a~starting vector $\mu^{(0)}\in\rk$ such that
\begin{equation}\label{eq_astart}
p_j(\mu_j^{(0)})q^{-(\mu_j^{(0)}\vee 0)} \le \nu_j(\mu^{(0)})\alpha \quad \text{for all } j=1,\ldots,m.
\end{equation}
Given $\mu^{(k)}$, define $\mu^{(k+1)}$ as solution of
\begin{equation}
\label{eq_step}
p_j(\mu_j^{(k+1)})q^{-(\mu_j^{(k+1)}\vee 0)} = \nu_j(\mu^{(k)})\alpha \quad \text{for all } j=1,\ldots,m.
\end{equation}
Stop if $\|\mu^{(k+1)} - \mu^{(k)}\|_2 < \varepsilon$ for some predefined threshold~$\varepsilon > 0$.
\end{algorithm}

By continuity and monotonicity, the first iterated value $\mu^{(1)}$ exists and is larger than (or equal to) $\mu^{(0)}$ component-wise (i.e.\ $\mu^{(1)}_j\geq\mu^{(0)}_j$ $\forall j=1,...,m$).  Since \eqref{eq_step} is satisfied and $\nu_j$ is non-decreasing, it follows that \eqref{eq_iteration} holds also for~$\mu^{(1)}$. By induction, we obtain that $(\mu^{(k)})_{k\ge 1}$ is a~non-decreasing sequence of vectors satisfying~\eqref{eq_iteration}. Because of the monotony of $\nu_j$, the sequence $(\mu^{(k)})_{k\ge 1}$ is bounded from above by $p_j^{-1}(\alpha_j)$, $j=1,\ldots,m$. Therefore, the sequence converges and the algorithm stops finally. We denote by $\mu^{\infty}$ the limiting value. 
In particular, if for some~$k\in\nn_0$ equality~\eqref{eq_step} holds for $\mu^{(k+1)}=\mu^{(k)}$, then the algorithm stops with $\mu^{\infty}=\mu^{(k)}$.

Our next goal is to show that $\mu^{\infty}$ is equal to the bounds (\ref{eq_isci}) of Definition~\ref{def_iSCI}. This results from the following theorem, which we proof in the appendix.

\begin{theorem} \label{thm_grenzwert}
The following properties are satisfied in Algorithm~\textup{\ref{alg_ki}}.
\begin{enumerate}
\item[\textup{(a)}] If $\mu^{(0)} \le L$ (component-wise), then also $\mu^{(k)}\le L$ for all $k\in\nn$.
\item[\textup{(b)}] The limiting value is independent of the starting value of the algorithm.
\item[\textup{(c)}] For any given starting value $\mu^{(0)}$, the limiting value satisfies $\mu^{\infty} \ge L$.
\end{enumerate}
We obtain from (a) to (c)  that $\mu^{(k)}$ converges to the confidence bounds  $L=(L_1,\ldots,L_m)$  defined in
(\ref{eq_isci}), if we find a valid starting value of Algorithm~\ref{alg_ki}.
\end{theorem}

%From Theorem~\ref{thm_grenzwert}(a) and (b) it follows that $\mu^{\infty}\le L$ for all starting values $\mu^{(0)}$, and hence by (c), we have $\mu^{\infty}= L$. %Therefore, , then the algorithm converges to~$L$. 

\begin{remark}[How to find a starting value.] \label{rem_startvalue}
%We show next that we can always find a starting value $\mu^{(0)}\le L$ which satisfies (\ref{eq_astart}).
%Note first that from the construction of the levels~$\alpha^\mu$, namely by rejecting the null hypotheses $H_j$ in the graph~$G^\mu$, it becomes clear that $\alpha_j^\mu \ge \alpha_j q^{\mu_j \vee 0}$. This is because the levels assigned to $H_j$ remain constant or increase when rejecting $H_k$ for $k\ne j$. 
As a starting value $\mu^{(0)}\le L$ we can choose $\mu_j^{(0)}=\min\{0,p_j^{-1}(\alpha_j)\}$, where we formally put $p_j^{-1}(0)=-\infty$. This vector obviously satisfies (\ref{eq_astart}) and $\mu^{(0)}\in [-\infty,0]^m$ with $\mu^{(0)}_j>-\infty$ whenever $\alpha_j>0$.
\end{remark}

%NB: Wenn $\alpha_j > 0$, muss der Startwert die Bedingung \eqref{eq_iteration} nicht erfüllen. Falls die umgekehrte Ungleichung gilt, ergibt sich monoton fallende Folge von Iterationswerten, die nach unten durch $p_j^{-1}(\alpha_j)$ beschränkt ist.
% Im Fall $\alpha_j = 0$ funktioniert der in R umgesetzte Algorithmus ebenfalls für jeden Startwert, was daran liegt, dass $p(x) = 0$ in R (auch wenn tatsächlich nur nahe bei Null), sofern $x$ klein genug ist (bisherige Wahl: x=-10). Für saubere Programmierung müssten alle Ergebnisse (abhängig von $\varepsilon$) gerundet werden. Eine untere Schranke ist dann jeder (endliche) Wert, für den der gerundete p-Wert gleich Null ist, d.h. der Algorithmus konvergiert immer.

\subsection{Properties of the confidence bounds} \label{sec_properties}

The main property of the new confidence bounds is that they are always informative in the sense of Definition~\ref{def_informative}. This is an essential advantage over the SCIs proposed by \citealp{SB08}.
\begin{proposition} \label{prop_informative}
The simultaneous confidence bounds obtained by Algorithm~\textup{\ref{alg_ki}} (or equivalently, by projection as described in Section~\textup{\ref{sec_definition}}) satisfy the conditions (a) and (b) of Definition~\textup{\ref{def_informative}}, i.e., they are always informative.
\end{proposition}

\begin{proof}
We start showing (a) of Definition~\ref{def_informative}. 
As discussed in Remark~\ref{rem_startvalue}, there exists a~starting value $\mu_j^{(0)} > -\infty$ if $\alpha_j >0$, hence $L_j > -\infty$ 
by Theorem~\ref{thm_grenzwert}.

If $\alpha_i = 0$, the starting value was set to $\mu_i^{(0)}=-\infty$. 
%The iterated values~$\mu_j^{(k)}$ remain to be~$-\infty$ until, 
If for some $k\ge 1$ and $j\not=i$, a positive level is shifted from $H_j$ to $H_i^{\mu_i^{(k)}}$ in the graph~$G^{\mu^{(k)}}$, then $\alpha_i^{\mu^{(k)}}=\nu_i(\mu^{(k)})\alpha>0$. By solving (\ref{eq_step}) in Algorithm~\ref{alg_ki} we obtain $\mu^{(k+1)}_i>-\infty$ and by (a) of Theorem~\ref{thm_grenzwert} that also $L_i>-\infty$. Hence, whenever at least one $H_j$, $j\not=i$, is rejected and, as a consequence, positive level is passed to $H_i$, then we will obtain $L_j>-\infty$.
This shows property (a).

We consider now property (b). Let $L$ be the vector of confidence bounds defined in (\ref{eq_isci}) for an arbitrary, but given data set. By Definition~\ref{def_iSCI} and Proposition~\ref{prop_nu} we have 
$p_j(L_j) \le \alpha_j^L = q^{L_j\wedge 0}\nu_j(L)\alpha$ for all $j=1,\ldots,m$. Assume now that, by a change in the data, the evidence against $H_i$ is increased while it remains constant or also increases for all $H_j$ with $j\not=i$. By our assumptions on the individual shifted p-values, this implies that $p_i(L_i)$ is strictly decreased while all other $p_j(L_j)$ remain constant or are decreased as well. Hence, with the new data, we obtain $p_i(L_i) < \alpha_i^L = q^{L_i\wedge 0}\nu_i(L)\alpha$ and $p_j(L_j) \le \alpha_j^L = q^{L_j\wedge 0}\nu_j(L)\alpha$ for all $j\not= i$.
As a consequence, $\mu^{(0)}=L$ can serve as starting point for the calculation of the new confidence bounds $L'$ for the changed data with Algorithm~\ref{alg_ki}. According to (\ref{eq_step}) and the monotonicity assumption of $p_i(\mu_i)$ in $\mu_i$, the first step of the algorithm results in $\mu^{(1)}_i>\mu^{(0)}_i=L_i$, and by (a) of Theorem~\ref{thm_grenzwert}, the new confidence bound $L'_i$ must be strictly larger than $L_i$. Therefore, $L_i$ increases with increasing evidence against $H_i$ when the evidence against the other $H_j$ remains the same or increases as well.
\end{proof}

To assure the informativeness of the SCIs, we pay a~price in terms of a~slightly reduced expected number of rejections to the underlying original graphical procedure. We can, however, control the desired power by the choice of the information weight~$q\in(0,1)$. Larger values of~$q$ are in favor of sharper confidence bounds, while smaller values of~$q$ lead to more rejections. This effect will be illustrated in more detail for some examples in Section~\ref{sec_examples}. The boundary case $q=1$ yields the weighted Bonferroni intervals with weights~$\alpha_1,\ldots,\alpha_m$.

One can also generalize the approach by choosing individual information weights $q_1,\ldots,q_m\in(0,1)$, depending on the importance of a~large bound~$L_j$ compared to rejecting as many hypotheses as possible. The method may even be applied with any positive, non-increasing, continuous function~$Q_j(\mu_j)$ (replacing $q^{\mu_j\vee 0}$) for $j=1,\ldots,m$ that is equal to~$1$ for $\mu_j\le 0$ and tends to~$0$ for $\mu_j\to\infty$. All arguments concerning the construction and properties of the SCIs work for functions~$Q_j(\mu_j)$ with these characteristics, as well.

%For easier reading, we consider only $q^{\mu_j\cap 0}$ with the same~$q$ for all hypotheses.

\begin{remark} \label{rem_extension}
Assumption~\ref{ass_complete} can be weakened so that the approach is also applicable for graphical procedures with non-complete graphs. We only need to adapt the transition weights in~$G^\mu$ for arrows from $H_j$ to $H_j^{\mu_j}$ if $\mu_j > 0$. Instead of~$q^{\mu_j}$, we define the weight
$$
1 - (1-q^{\mu_j})\sum_{i=1}^m g_{ji},
$$
which reduces to $q^{\mu_j}$ if the original graph is complete. It can be easily seen that~$G^\mu$ is a~complete graph so that \eqref{eq_alpha} is satisfied. The expression~$q^{-\mu_j\vee 0}$ in Algorithm~\ref{alg_ki} has then to be modified accordingly. As stated above, all arguments concerning the properties of the resulting SCIs remain valid with this modified weight function.

Alternatively to this modification, one could of course also add arrows in the original graph so that Assumption~\ref{ass_complete} is satisfied. This can be done if one wants to increase power for certain hypotheses rather than improve confidence assertions. However, the latter is not always possible, for example in hierarchical testing.
\end{remark}

\section{Examples} \label{sec_examples}

\subsection{Bonferroni-Holm procedure} \label{sec_holm}

The graph in Figure~\ref{fig_graphs}(a) represents the weighted Bonferroni-Holm procedure for three hypotheses. If $g_{jk}=1/2$ for $j,k=1,2,3$, $j\ne k$, then the unweighted Holm test evolves. A~generalization for $m$ instead of $3$~hypotheses is straightforward. The simultaneous confidence intervals of \citealp{SB08} compatible with this procedure have the drawback that the bounds for rejected hypotheses are only informative if all hypotheses are rejected. Simple SCIs are possible for the Bonferroni procedure which, however, is not as powerful as the Holm procedure. The SCIs presented here propose a~compromise, which is more powerful than Bonferroni and produces always informative confidence bounds. 

The same features hold for the approach of penalized simultaneous confidence intervals introduced in \citealp{BS14} where attention is restricted to the unweighted Holm procedure and specific union intersection tests. The penalized SCIs for the Holm test are constructed via dual weighted Bonferroni tests, where the weights are based on a~so-called ``penalization function''. The penalization function $\lambda_i(\mu_i) = \exp(a \mu_i\vee 0)$ was proposed, with $a>0$ as adjusting parameter for the interpolation between the importance of sharp confidence bounds ($a=0$ corresponds to Bonferroni) and high power ($a\to\infty$ corresponds to Holm). The information weight~$q\in(0,1)$ of the new approach has a~similar meaning as $-\log(a)$. Due to the different weighting scheme, the penalized SCIs are not exactly equal to the intervals introduced here. However, their properties are closely related, as illustrated by a simulation next.

The scenario of the simulation is the following: Five null hypotheses (e.g., the effects of different treatments compared to placebo) are tested with equal weights. We assume a~scenario where all true effects are equally large. Other scenarios led to similar results. The significance level is $\alpha=2.5\%$ and the standard errors in (\ref{eq_spval}) are assumed to be $1$. We further assume that the study is powered such that the probability to reject any individual hypothesis at significance level~$\alpha/5$ is $80\%$. We made $10000$ simulation runs and determined the confidence bounds for both approaches with several values of the parameters~$q$ and~$a$, respectively. Figure~\ref{fig_penalized} compares the two procedures with respect to their trade-off between the mean confidence bound -- which is high for larger information weight~$q$ resp.\ smaller~$a$, and the average number of rejected hypotheses -- which is high for smaller~$q$ resp.\ larger~$a$. We see that the trade-off curves are almost the same for both approaches.

\begin{figure}
\begin{center}
			\includegraphics[scale=0.8]{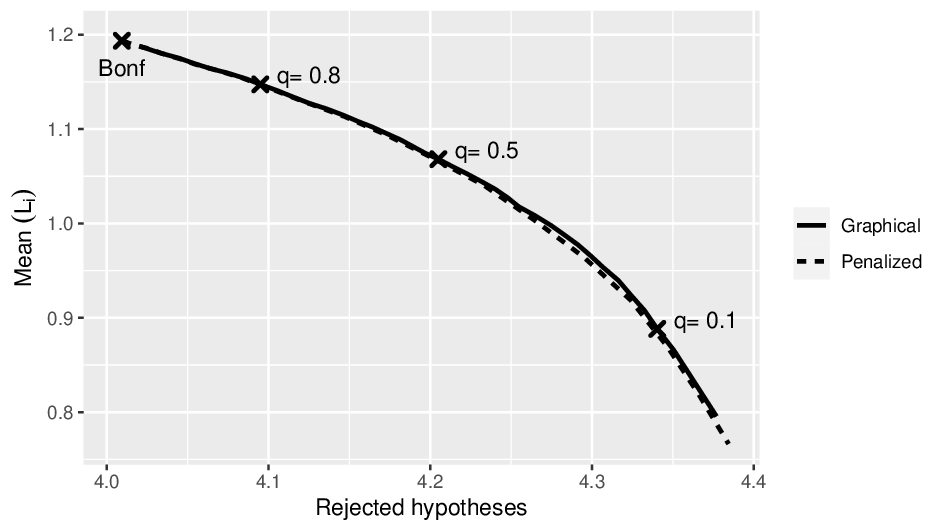}
\end{center}
      \caption{Trade-off between mean confidence bound~$L_i$ (if $H_i$ is rejected) and average number of rejected hypotheses. Comparison between the new SCIs for graphical procedures and the penalized SCIs with exponential penalizing function in dependence of different values of the information weight~$q$ resp.~$a$. The true effects are $\theta_i=c$, $i=1,\ldots,5$, with $c>0$ such that the Bonferroni test has Power $80\%$ for each separate hypothesis at significance level~$\alpha=0.025$. The value~$q=1$ corresponds to the Bonferroni intervals.}
  \label{fig_penalized}
\end{figure}

\subsection{Fixed sequence and fallback procedure}

In \citealp{SB14,SB15}, we have introduced informative simultaneous confidence intervals for the hierarchical and the fallback procedure, respectively. We show in the Appendix that Algorithm~\ref{alg_ki} produces exactly the same SCIs. We have discussed properties of these SCIs in our previous works. In particular, the SCIs are always informative. This holds for all graphical informative SCIs, as we have shown in Proposition~\ref{prop_informative}. The case of gatekeepers discussed in Definition~\ref{def_informative} and Remark~\ref{rem_startvalue} is of particular relevance for the fixed sequence: If some $H_i$ with $i < j$ is accepted, then no level is shifted to test~$H_j$ and hence $L_j = -\infty$.
 
A~nice feature of the informative SCIs is that they somehow respect the ordering in the hierarchical and the fallback procedure, in the sense that there is no power loss compared to the original procedure for the first hypothesis, which is normally the most important one. The price that has to be paid for more information is a slight power loss for $H_2,\ldots,H_m$. As for the more general informative SCIs introduced here, one has some control over the information--power trade-off by the choice of the information weight~$q$ (see Section~\ref{sec_properties}).

%\subsection{A more complex graphical procedure} \label{sec_example}
\subsection{A clinical trial example}  \label{sec_example}

We consider now a hypothetical clinical trial example that is in line with the RELY trial, reported in \citealp{C09}. In this trial two doses of the thrombin inhibitor dabigatran were compared to warfarin (active control) in a randomized and semi-blinded, multi-arm clinical trial. The primary treatment goal is the risk reduction of strokes or systematic embolisms in patients with atrial fibrillation. There are a primary efficacy and a safety parameter (both hazard ratios). The data indicated that the lower dabigatran dose is non-inferior and the higher dose even superior to walfarin with regard to efficacy (i.e.\ hazard for a stroke or systematic embolism), and that both dabigatran doses seem to be superior to walfarin with regard to safety (hazard for major bleeding). Since the multiplicity adjustments anticipated in the trial were only with regard to the two non-inferiority null hypotheses in efficacy, superiority claims with regard to efficacy and claims on safety are not strictly confirmative. Informative simultaneous confidence intervals based on as graph as in Figure~\ref{fig_graph},
but with two doses instead of three, permits strictly confirmative claims also with regard to superiority in the efficacy and safety endpoints. We will illustrate below how the study could have been planned to include such a procedure. We will assume, for simplicity, that for the estimates for efficacy and safety are normally distributed with known standard errors (which is in line with the common asymptotic approximation for the log hazard rates). 

As indicated in Figure~\ref{fig_graph}, effectiveness of the doses is primarily tested by non-inferiority. Accordingly, the hypotheses are $H_{E_j} = H_j^{-\delta_j}:\theta_{E_j}\le - \delta_j$ with given non-inferiority margins $\delta_j > 0$ for $j=1,2$, where $\theta_{E_j}$ are the efficacy parameters. Here $j=1$ represents the low dose treatment and $j=2$ the high dose treatment. After non-inferiority of a~treatment $j$ has been shown, its safety is investigated by the~superiority hypothesis $H_{S_j}$, $j=1,2$, for major bleeding. 

Of course, it would be valuable to also show superiority over the active comparator for the efficacy endpoint if the effect estimate is large. To this end, the simultaneous confidence bounds would need to be informative, such that $L_j > 0$ is a possibility whenever non-inferiority has been shown.
%This test is not included in the graph, but superiority would be confirmed if $L_j > 0$. 
With the new approach, we obtain informative confidence bounds for all hypotheses and thus also have the chance to prove superiority. In contrast, the SCIs of \citealp{SB08} would give e.g.\ $L_1 = -\delta_1$ if $H_{E_1}$ is rejected but not $H_{E_2}$. 

We now discuss how the trial could have been planned with our informative simultaneous confidence intervals. To this end we fix the one-sided overall level $\alpha=0.025$ and assign the two non-inferiority null hypotheses the initial levels $\alpha_1=\alpha_2=\alpha/2=0.0125$. The non-inferiority margin for the log-hazard rate is chosen, like in the RELY study, as $\delta_1=\delta_2=\delta_n:=\log(1.46)=0.378$ for both doses. For a power of $80\%$ (with the initial levels) we need to recruit patients until the total information $I_n=(z_{\alpha_j}+z_\beta)^2/\delta_n^2=66.37$ is reached for the efficacy endpoint, where $z_u$ is the $(1-u)$-quantile of the standard normal distribution. Assuming that superiority in efficacy is powered at $\delta_e=1.3\delta_n=0.492$ (compare \citealp{BS14}), we can assign the smaller level $\alpha_e=1-\Phi\big(\delta_e I_n - z_\beta\big)=0.00077$ to achieve the power $80\%$ with the information $I_n$. From this we derive the corresponding $q=q_{E_j}$ for the calculation of the SCIs for the hypotheses $H_{E_j}$, $j=1,2$. We have:
$$\alpha_{E_j}=q_{E_j}^{0+\delta_n}\alpha/2\Leftrightarrow q\approx 0.00063.$$ With this choice of $q_{E_j}$ we ensure that $L_{E_j}\geq 0$ with probability $80\%$ for $\theta_{E_j}=\delta_e$.
The values of the information weight $q$ for the two safety hypotheses $H_{S_j}$, $j=1,2$ remain to be chosen. We want to explore the behavior of the SCI bounds in dependence of these parameters in a simulation. Given the symmetry of the study design, an equal choice of the information weight~$q$ for both dose groups seems reasonable. Gaining knowledge of the size of the effect in the safety endpoints is clearly most important, in case of no superiority in the efficacy endpoints. This leads to the scenario where $\theta_{E_1}=\theta_{E_2}=0$, $\theta_{S_1}=\theta_{S_2}=\delta_e=0.492$. We want to investigate the influence of $q_{S_1}=q_{S_2}$ and vary it over $(0,1)$. For each choice of $q_{S_1}=q_{S_2}$ we simulated 100.000 trials with the above parameters. For each of the simulation replicas, we calculated the informative SCIs as well as the compatible SCIs by \citealp{SB08}. Since in \citealp{C09} no correlation between the test statistics for $H_{E_j}$ and $H_{S_j}$, $j=1,2$ was given, we set this to 0. Following \citealp{DSG11}, the correlation between the test statistics for $H_{E_j}$, $j=1,2$ as well as the correlation between the tests statistics for $H_{S_j}$, $j=1,2$ were each set to $1/2$. The resulting relation between the probability to reject $H_{S_j}$ and the magnitude of the confidence bound $L_{S_j}$ is shown in Figure~\ref{fig_simout}.

\begin{figure}%
\begin{center}
\includegraphics[scale=0.8]{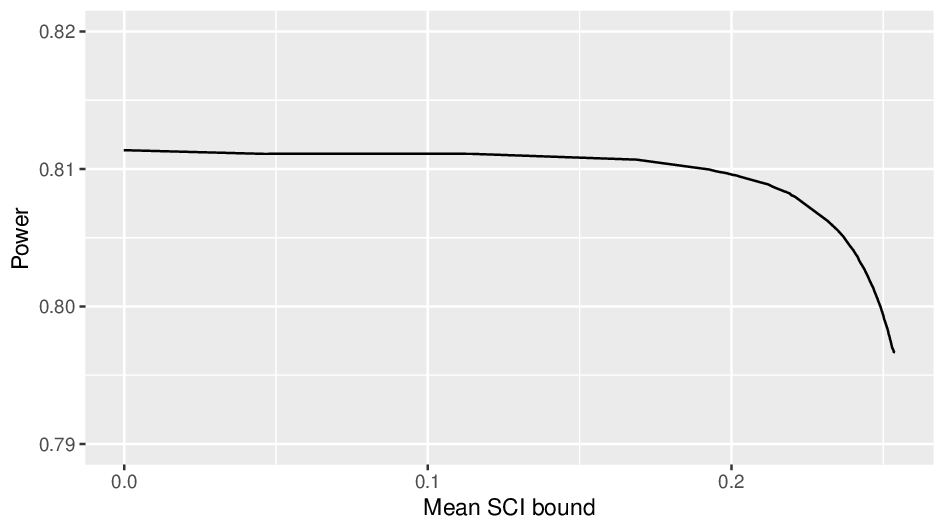}%
\end{center}
\caption{Probability to reject $H_{S_j}$ over the mean informative SCI bound for $H_{S_j}$ based on the simulation to calibrate the information weight~$q_{S_j}$}%
\label{fig_simout}%
\end{figure}

For the simulation we used our developed \textsf{R}-package \emph{informativeSCI} which is available from \textsf{CRAN} or \textsf{github}. The package can be used to run Algorithm~\ref{alg_ki}, i.e. for the calculation of the lower informative confidence bounds. Additionally it can be used to explore the behavior of the SCI bounds in dependence of the information weights by performing simulations.

From the simulation results we can do the trade-off between higher power and higher information of the SCIs. For example we could choose that value of the information weight~$q$ that ensures, that the probability to reject $H_{S_j}$ with the informative SCIs is still $81\%$ (which amounts to $q_{S_j}=10^{-10}$ in this setup) or $80\%$ (which amounts to $q_{S_j}=0.38$ in this setup). For planning a specific clinical trial, one would then investigate the power and magnitude of the SCI bounds for those values of $q$ under different scenarios and determine the value of $q_{S_j}$ to be used in the trial based on the overall performance of the SCIs in those scenarios. Results of such considerations are given in Table~\ref{tab_simres_scenarios} in the appendix. We observe, that the loss in power compared to the compatible SCIs is rather small in all scenarios, where the information gain can be quite substantial. The choice of $q=10^{-10}$ seems to be favorable, given that the power loss for $q=0.38$ amounts to up to $4\%$, depending on the scenario.

\section{Discussion} \label{sec_discussion}

The approach of informative graphical SCIs proposes a~way out of the conflict between wishing to reject as many hypotheses as possible and at the same time obtaining relevant information on the parameters of interest by confidence intervals. The SCIs can be constructed for all graphical procedures from \citealp{B09} and will always give more information than the pure (non-)rejection of the hypotheses, except in the case where gatekeepers are not rejected and thus the hypothesis test and interval estimation of the parameter is considered to be of no interest. As usual, there is no gain without costs. Compared to the original graphical procedures, the informative SCIs usually pay a~prize in terms of rejections. The~possibility to adapt the SCIs to trial specific priorities concerning information and/or power is via the choice of the information weight~$q$. High values of~$q$ (i.e., close to~$1$) give more informative SCIs, while small values (close to~$0$) lead to a~higher probability to reject more hypotheses. The choice of the information weight~$q$ can also be made separately for each hypothesis. An example how~$q$ can be determined in a~clinical trial has been given in Section~\ref{sec_example}.

\paragraph{Extension to $\varepsilon$-Graphs.} So far, we have not explicitly considered graphs that use the $\varepsilon$-notation introduced in \citealp{B09}, which makes it possible to shift level between families of hypotheses. For example, one may define a~procedure where the significance level~$\alpha$ is first shifted between~$H_1$ and $H_2$, and only if both hypotheses are rejected, the level is passed to~$H_3$ (see Figure~9 in \citealp{B09}). When applying our procedure to this graph, the modified graph contains the shifted hypotheses $H_1^{\mu_1}$ and~$H_2^{\mu_2}$. Since these hypotheses remain after rejecting the null hypotheses, the level kept for $H_1$ and~$H_2$ will never be shifted to~$H_3$. Hence, the resulting SCIs will always be the same as if $\varepsilon$ equals zero. This means that the construction of SCIs comes at the cost of not being able to exploit the improvements in power that the introduction of $\varepsilon$-edges may yield. It is certainly possible to extenuate this effect by modifying the graph~$G^\mu$ appropriately.

\paragraph{Restrictions and outlook.} As we have seen in Section~\ref{sec_holm}, the new method is more flexible than the penalized intervals from \citealp{BS14} in that all graphical procedures of \citealp{B09} can be adapted to informative graphical SCIs. However, it has to be recognized that an advantage of our previous approach was that it could be generalized to other step-down tests like the Dunnett procedure, which accounts for correlations between the test statistics and thus gains in power. A~future aspect of work may therefore be to incorporate correlations also in the informative graphical SCIs, similar to the work of \citealp{B11} for the graphical procedure.

\paragraph{Acknowledgement.} We thank Serhat G\"unay for his help in the development of R programs for our simulation study. This research was supported by the DFG, grant BR 3737/1-1.

\bibliographystyle{apalike}
\bibliography{main} 

\begin{thebibliography}{}

\bibitem[Bauer et~al., 2001]{B01}
Bauer, P., Brannath, W., and Bretz, F. (2001).
\newblock Multiple testing for identifying effective and safe treatments.
\newblock {\em Biometrical Journal}, 43:605--616.

\bibitem[Brannath and Schmidt, 2014]{BS14}
Brannath, W. and Schmidt, S. (2014).
\newblock A new class of powerful and informative simultaneous confidence
  intervals.
\newblock {\em Statistics in Medicine}, 33(19):3365--86.

\bibitem[Bretz et~al., 2009]{B09}
Bretz, F., Maurer, W., Brannath, W., and Posch, M. (2009).
\newblock A graphical approach to sequentially rejective multiple test
  procedures.
\newblock {\em Stat. Med.}, 28(4):586--604.

\bibitem[Bretz et~al., 2011]{B11}
Bretz, F., Posch, M., Glimm, E., Klinglmueller, F., Maurer, W., and Rohmeyer,
  K. (2011).
\newblock Graphical approaches for multiple comparison procedures using
  weighted {B}onferroni, {S}imes, or parametric tests.
\newblock {\em Biom. J.}, 53(6):894--913.

\bibitem[Connolly et~al., 2009]{C09}
Connolly, S., Ezekowitz, M., Salim, Y., Eikelboom, J., Oldgren, J., Parekh, A.,
  Pogue, J., Reilly, P., Themeles, E., Varrone, J., Wang, S., and Alings, M.
  (2009).
\newblock Dabigatran versus warfarin in patients with atrial fibrillation.
\newblock {\em N Engl J Med}, (361):1139--1151.

\bibitem[Di~Scala and Glimm, 2011]{DSG11}
Di~Scala, L. and Glimm, E. (2011).
\newblock Time-to-event analysis with treatmentarm selection at interim.
\newblock {\em Statistics in Medicine}, (30):3067--3081.

\bibitem[Dickhaus, 2014]{D14}
Dickhaus, T. (2014).
\newblock {\em Simultaneous Statistical Inference}.
\newblock Springer Berlin, Heidelberg.

\bibitem[European Medicines~Agency, 2017]{EMA17}
European Medicines~Agency, . (2017).
\newblock Guideline on multiplicity issues in clinical trials.
\newblock Technical Report EMA/CHMP/44762/2017, European Medicines Agency.

\bibitem[Guilbaud, 2008]{G08}
Guilbaud, O. (2008).
\newblock Simultaneous confidence regions corresponding to {H}olm's step-down
  procedure and other closed-testing procedures.
\newblock {\em Biom. J.}, 50(5):678--692.

\bibitem[Hochberg and Tamhane, 1987]{HT87}
Hochberg, E. and Tamhane, A.~C. (1987).
\newblock {\em Multiple Comparison Procedures}.
\newblock Wiley.

\bibitem[Schmidt and Brannath, 2014]{SB14}
Schmidt, S. and Brannath, W. (2014).
\newblock Informative simultaneous confidence intervals in hierarchical
  testing.
\newblock {\em Methods of Information in Medicine}, 53:278--283.

\bibitem[Schmidt and Brannath, 2015]{SB15}
Schmidt, S. and Brannath, W. (2015).
\newblock Informative simultaneous confidence intervals for the fallback
  procedure.
\newblock {\em Biometrical Journal}, 57(4):712--719.

\bibitem[Strassburger and Bretz, 2008]{SB08}
Strassburger, K. and Bretz, F. (2008).
\newblock Compatible simultaneous lower confidence bounds for the {H}olm
  procedure and other {B}onferroni-based closed tests.
\newblock {\em Stat. Med.}, 27(24):4914--4927.

\end{thebibliography}

\appendix
\section{Proof of \eqref{eq_alpha}}
We will explain that \eqref{eq_alpha} holds for any starting graph~$G$. First, as argued in the main article, the modified graph~$G^\mu$ is complete with local levels summing up to~$\alpha$ by construction and because of Assumption~1. How can level get lost while rejecting the null hypotheses in~$G^\mu$? The answer can be derived from~\eqref{update}. If all updated transition weights are computed by the first case of~\eqref{update}, then no level gets lost. Indeed, it can be easily calculated that $\sum_{l=1}^m \alpha_l^{\mathrm{new}} = \alpha$ and $\sum_{l\ne j}g_{jl}^{\mathrm{new}} = 1$ if $\sum_{l=1}^m \alpha_l = \alpha$ and $\sum_{l\ne j}g_{jl} = 1$. The only way to loose level is a~situation where $H_i$ is rejected and for some~$j$ the updated transition weights $g_{jl}^{\mathrm{new}}$ are given by the second case of~\eqref{update}, i.e., equal zero, because $g_{ji}g_{ij} = 1$. Intuitively, $H_i$ and $H_j$ form a~loop spending all level to each other without connection to the other hypotheses. If then both $H_i$ and $H_j$ are rejected, their level cannot be transferred further. 

In the modified graph~$G^\mu$, however, every hypothesis~$H_i$ is connected to its shifted hypothesis~$H_i^{\mu_i}$ with the positive transition weight~$q^{\mu_i}$. Since the shifted hypotheses are not rejected, the above situation cannot occur. If two null hypotheses form a~loop with no connection to another null hypothesis and both are rejected, then all remaining level will be transferred to their respective shifted hypotheses. Therefore, no level can get lost and \eqref{eq_alpha} is satisfied.

\section{Proof of Theorem~\ref{thm_grenzwert}}
In this section, we show that the limiting values of Algorithm~\ref{alg_ki} are equal to the confidence bounds defined in Section~\ref{sec_definition}. We start by deriving important properties of the limiting values that are needed for the proof (Section~\ref{sec_properties}). In Sections \ref{sec_proofa}--\ref{sec_proofc}, we prove the three parts of Theorem~\ref{thm_grenzwert}, which together yield the equivalence. 

\subsection{Important properties of $\mu^{\infty}$} \label{sec_properties}

We first remind that the individual p-values $p_j$ are continuous by assumption and the weight functions~$\nu_j$ are continuous by Proposition~\ref{prop_nu}. Because Equation~\eqref{eq_step} holds for all~$k$ and $\mu^{(k)} \to \mu^{\infty}$, we obtain the equality 
$$
p_j(\mu_j^{\infty})q^{-(\mu_j^{\infty}\vee 0)} = \nu_j(\mu^{\infty})\alpha \quad \text{for all } j=1,\ldots,m.
$$

By Proposition~\ref{prop_nu}, it follows that
\begin{equation} \label{eq_muinfty}
p_j(\mu_j^{\infty}) = \alpha_j^{\mu^{\infty}}\quad \text{for all } j=1,\ldots,m.
\end{equation}
We reject an intersection hypothesis $H^{\mu}$ if and only if one of its components $H_j^{\mu_j}$ is rejected at local level~$\alpha_j^{\mu}$. Therefore, the adjusted p-value for $H^{\mu}$ is
\begin{equation} \label{eq_wj}
p(\mu) = \min_{\substack{j=1,\ldots,m\\ w_j(\mu)>0}}\frac{p_j(\mu_j)}{w_j(\mu)}, \quad \text{where }
w_j(\mu) = \frac{\alpha_j^{\mu}}{\alpha} = q^{\mu_j\vee 0}\nu_j(\mu).
\end{equation}
By \eqref{eq_muinfty}, 
\begin{equation} \label{eq_setM}
\mu^{\infty} \in M:=\Big\{\mu\in\rk|\exists u\in (0,1): p_j(\mu_j)=u\cdot w_j(\mu),~\forall j=1,...,m\Big\}
\end{equation}
and
\begin{equation} \label{eq_peqalpha}
p(\mu^{\infty}) = \alpha.
\end{equation}

\subsection{Proof of Part (a) of Theorem~\ref{thm_grenzwert}} \label{sec_proofa}

In the following, we use the notation $x \le y$ if $x,y\in\rk$ for the component-wise ordering, i.e. $x_j\le y_j$ for $j=1,\ldots,m$. Let $\mu^{(0)}\le L$. Theorem~\ref{thm_grenzwert}(a) states that this implies $\mu^{(k)}\le L$ for all $k\in\nn$. We will show this assertion by induction.

Assume that $\mu^{(k)} \le L$ for some $k\in\nn_0$. For $\delta\in\rk$, we define the half-space
\begin{equation*}
H(\delta):=\bigcup_{j=1}^m\{\mu\in\rk:\mu_j \le \delta_j\}.
\end{equation*}
Because the SCI $\bigtimes _{j=1}^m (L_j,\infty)$ contains the confidence set 
\begin{equation} \label{confset}
C=\{\mu\in\rk:H^{\mu} \text{ is not rejected}\}=\{\mu\in\rk:p(\mu) > \alpha\},
\end{equation}
it becomes clear that
\begin{equation}
 \label{eq_halbraum}
\delta_j\le L_j \text{ for } j=1,\ldots,m \iff p(\mu) \le \alpha \text{ for all } \mu \in H(\delta).
\end{equation}

We have argued in the main article that Algorithm~\ref{alg_ki} produces a~non-decreasing sequence, i.e., $\mu^{(k+1)} \ge \mu^{(k)}$, hence we have $H(\mu^{(k)})\subseteq H(\mu^{(k+1)})$. By induction assumption and \eqref{eq_halbraum}, we conclude that $p(\mu)\le\alpha$ for all $\mu\in H(\mu^{(k)})$ and have to show the same property for all values in the difference set
\begin{align*}
R_k &= H(\mu^{(k+1)})\setminus H(\mu^{(k)}) \\
  &= \bigcup_{j=1}^m\{\mu\in\rk: \mu > \mu^{(k)} \text{ and } \mu_j \le \mu_j^{(k+1)}\}.
\end{align*}
Let $\mu\in R_k$. Then $\mu > \mu^{(k)}$ and there exists a~$j$ such that $\mu_j \le \mu_j^{(k+1)}$. By monotonicity of~$p_j$ and~\eqref{eq_step},
$$
p_j(\mu_j)q^{-(\mu_j\vee 0)} \le p_j(\mu_j^{(k+1)})q^{-(\mu_j^{(k+1)}\vee 0)} = \nu_j(\mu^{(k)})\alpha
\le \nu_j(\mu)\alpha.
$$ 
Hence, $p(\mu)\le\alpha$, which concludes the induction.

\subsection{Proof of Part (b) of Theorem~\ref{thm_grenzwert}} \label{sec_proofb}

To show that the limit value $\mu^{\infty}$ of Algoritm~\ref{alg_ki} is independent of the starting value~$\mu^{(0)}$, we need the following lemma.

\begin{lemma} \label{lemma_summe}
Let $\mu\in M$ with $M$ defined in \eqref{eq_setM}. Then
$$
p(\mu) = \sum_{j=1}^m p_j(\mu_j).
$$
\end{lemma}

\begin{proof}
Let $\mu\in M$ with $M$ defined in \eqref{eq_setM}. Then there is a $u\in(0,1)$ with $p_j(\mu_j)=w_j(\mu)\cdot u$ for all $j=1,...,m$. In particular, for all $j$ with $w_j(\mu)\neq 0$ we get $p_j(\mu_j)/w_j(\mu)=u$. By the definition of $p(\mu)$ in \eqref{eq_wj} we get $p_j(\mu_j)/w_j(\mu)=p(\mu)$ for all $j$ with $w_j(\mu)\neq 0$. Hence:

$$\sum_{j=1}^mp_j(\mu_j)=\sum_{\substack{j=1\\ w_j(\mu)\neq 0}}^mp_j(\mu_j)=\sum_{\substack{j=1\\ w_j(\mu)\neq 0}}^m\frac{p_j(\mu_j)}{w_j(\mu)}w_j(\mu)=p(\mu)\sum_{\substack{j=1\\ w_j(\mu)\neq 0}}^mw_j(\mu)=p(\mu).$$
\end{proof}

We now show part (b) of Theorem~\ref{thm_grenzwert}. Assume there exist $\mu,\mu^\prime\in M$ with $\mu\neq\mu^\prime$ and $p(\mu)=p(\mu^\prime)=\alpha$. We define $\tau=(\tau_1,...,\tau_m)$ by $\tau_j=\max\{\mu_j,\mu^\prime_j\}$. Since $\mu\in M$ and $\mu\leq\tau$ we have for all $j$ (w.l.o.g. let $\tau_j=\mu_j)$: 
\begin{eqnarray*}
p_j(\tau_j)=p_j(\mu_j)=w_j(\mu)\alpha=q^{\mu_j\vee 0}\nu_j(\mu)\alpha=q^{\tau_j\vee 0}\nu_j(\mu)\alpha\leq q^{\tau_j\vee 0}\nu_j(\tau)\alpha=\alpha_j^\tau
\end{eqnarray*}

Therefore $\tau$ is a valid starting value for Algorithm~\ref{alg_ki}. The algorithm then converges to some $\tau^\infty\in M$ (cf.~\eqref{eq_setM} and \eqref{eq_peqalpha}). Since $\mu\leq\tau$ and $\mu^\prime\leq\tau$ and $\mu\neq\mu^\prime$ there is at least one index $1\leq l\leq m$ with $\mu_l<\tau_l$ or $\mu^\prime_l<\tau_l$. W.l.o.g. assume that $\mu^\prime_l<\tau_l$. Since the sequence $\tau^{(k)}$ generated by Algorithm~\ref{alg_ki} is increasing in each component, we also have $\mu^\prime\leq\tau^\infty$ and $\mu^\prime_l<\tau^\infty_l$. By the strict monotonicity of the p-values we obtain $$\sum_{j=1}^mp_j(\mu^\prime_j)<\sum_{j=1}^mp_j(\tau_j^\infty).$$ With \eqref{eq_peqalpha} and Lemma~\ref{lemma_summe} it follows $$\alpha=p(\mu^\prime)=\sum_{j=1}^mp_j(\mu^\prime_j)<\sum_{j=1}^mp_j(\tau_j^\infty)=p(\tau^\infty)=\alpha.$$ Because of this contradiction, there can be at most one $\mu\in M$ with $p(\mu)=\alpha$. It can thus not depend on the starting value.

\subsection{Proof of Part (c) of Theorem~\ref{thm_grenzwert}} \label{sec_proofc}

From condition~\eqref{eq_astart} and the fact that $\nu_j(\mu)\in\mathbb R_{\geq 0}$, we see that any starting point of the algorithm is also a~starting point when the significance level is $u\ge \alpha$. We can therefore define a mapping
$$
\varphi\colon [\alpha,1) \to M,
$$
where $\varphi(u)=(\varphi_1(u),\ldots,\varphi_m(u))$ is the limit point for the algorithm with significance level~$u$. Note that this mapping does not depend on the starting point (as long as the starting point satisfies \eqref{eq_astart}), as we have shown in the previous section.

\begin{proposition} \label{prop_m}
The mapping~$\varphi$ is
\begin{enumerate}
\item [\textup{(a)}] strictly increasing, i.e., $u_1 < u_2$ implies $\varphi(u_1)\le \varphi(u_2)$ component-wise with $\varphi(u_1)\ne \varphi(u_2)$ %(i.e. $\varphi(u_1)_j < \varphi(u_2)_j$ for some~$j$),
\item [\textup{(b)}] continuous, i.e., for any sequence $u_n \to u$, we have $\varphi(u_n) \to \varphi(u)$.
\end{enumerate}
\end{proposition}

\begin{proof}
\begin{enumerate}
\item [\textup{(a)}] If $u_1 < u_2$, then by the monotonicity of~$\nu_j$
$$
p_j(\varphi_j(u_1))q^{-(\varphi_j(u_1)\vee 0)} =  \nu_j(\varphi(u_1))u_1 
\le \nu_j(\varphi(u_1))u_2,
$$
hence $\varphi(u_1)$ is a~starting value for the algorithm at level~$u_2$, which converges to $\varphi(u_2)$. By Part~(a) of Theorem~\ref{thm_grenzwert}, $\varphi(u_1)\le \varphi(u_2)$, and by \eqref{eq_peqalpha}, $$p(\varphi(u_1))=u_1<u_2=p(\varphi(u_2)),$$  hence, $\varphi(u_1)\ne \varphi(u_2)$.
\item [\textup{(b)}] 
To show (b) we show that $\phi$ is right continuous and left continuous in every point. Let $u\in (0,1)$ and $(u_n)_{n}$ be a monotonously decreasing sequence in $(u,1)$ with $u_n\to u$ for $n\to\infty$. By (a) the sequence $(\varphi_j(u_n))_n$ is also monotonously decreasing and convergent to some finite value (if bounded) or to $-\infty$ (if not bounded). We denote the limiting value by $\varphi_j^{+}$. Then the limiting value of $(\varphi(u_n))_n$ is given by $\varphi^{+}:=(\varphi_1^{+},...,\varphi_m^{+})$. We will now show that $\varphi^+=\varphi(u)$. By the continuity of $w_j$ and $p_j$ (cf.~Proposition~\ref{prop_nu}) \begin{equation}w_j(\varphi(u_n))\overset{n\to\infty}{\longrightarrow}w_j(\varphi^+)\quad\text{and}\quad p_j(\varphi_j(u_n))\overset{n\to\infty}{\longrightarrow}p_j(\varphi_j^+).\label{eq_contwp}\end{equation} 
Because $\varphi(u_n)$ is a limiting value of Algoritm~\ref{alg_ki} and by \eqref{eq_muinfty} we have that 
$$p_j(\varphi_j(u_n))=w_j(\varphi(u_n))u_n.$$ By \eqref{eq_contwp} and the convergence of $u_n$ we obtain $$p_j(\varphi_j^+)=w_j(\varphi^+)u.$$ Hence, $\varphi^+\in M$ and additionally, by \eqref{eq_wj} $p(\varphi^+)=u$. By the argumentation in the proof of Part (b) of Theorem~\ref{thm_grenzwert} we know that there is at most one value $\mu\in M$ with $p(\mu)=u$. Therefore, $\varphi^+=\varphi(u)$ and $\varphi$ is shown to be right continuous. By a similar argumentation, $\varphi$ can also shown to be left continuous. Therefore, $\varphi$ is continuous.
\end{enumerate}
\end{proof}

We can now prove that $\mu^\infty \ge L$. Let $u_n$ be a~sequence in $(\alpha,1)$ that converges to~$\alpha$. Then $p(\varphi(u_n)) = u_n > \alpha$ and hence $\varphi(u_n)\in C$ for all~$n$, see~\eqref{confset}. By definition of~$L$, we conclude that $\varphi(u_n) \ge L$. It follows from Proposition~\ref{prop_m} that $\varphi(u_n)$ converges to $\varphi(\alpha)=\mu^{\infty}$ and therefore $\mu^{\infty} \ge L$.

\section{Equivalence of new SCIs and former approach for the fallback procedure} \label{sec_fallback}

We show in this Section that Algorithm~1 produces the same SCIs as Algorithm~1 of \cite{SB15}, where informative SCIs for the fallback procedure are proposed. Denote by $c_j=\alpha_j/\alpha$ the initial weights of the fallback procedure. The fixed-sequence test is contained as the special case, where $\alpha_1=1$ and $\alpha_j=0$ for $j=2,\ldots,m$. Therefore, also the informative SCIs for the hierarchical procedure introduced in \cite{SB14} are contained in the graphical SCI approach.

Consider Algorithm~1 with the function
\begin{equation}
\label{eq_fallback}
\nu_j(\mu) = \sum_{l=1}^{j-1}c_l\prod_{s=l}^{j-1}\big(1-q^{\mu_s\vee 0}\big)+  c_j.
\end{equation}
Further, we have to replace the term $q^{\mu_j\vee 0}$ in Algorithm~1 by the function
$$
\omega_j(\mu_j)=\begin{cases}
			 q^{\mu_j\vee 0} & j<m,\\
			 1						   & j=m,
			 \end{cases}
$$
due to the fact that the fallback procedure is not complete (see Remark~\ref{rem_extension}). One can show that $\nu_j(\mu)$ in~\eqref{eq_fallback} is indeed the function defined by Proposition~\ref{prop_nu} for the modified graph~$G^\mu$ of the fallback procedure. With this notation, Algorithm~1 is equivalent to Algorithm~1 in \cite{SB15}.

Indeed, for any starting value~$\mu^{(0)}$, we have $\nu_1(\mu^{(0)}) = c_1$ in~\eqref{eq_fallback}. Then Condition~\eqref{eq_step} for $j=1$ translates to the two cases described in Algorithm~1 of \cite{SB15}: If $p_1(0) > c_1\alpha$, then we find $\mu_1^{(1)} < 0$ such that $p_1(\mu_1^{(1)}) = c_1\alpha$. If $p_1(0) \le c_1\alpha$, then we find $\mu_1^{(1)} \geq 0$ such that $p_1(\mu_1^{(1)})q^{-\mu_1^{(1)}} = c_1\alpha$. In the second iteration step, $\nu_1(\mu^{(1)})=c_1$ remains unchanged in~\eqref{eq_fallback}. Therefore, $\mu_1^{(k)} = L_1$ will not change further. But $\nu_2(\mu^{(1)}) = c_1(1-q^{L_1\vee 0}) + c_2$, which is exactly the weight $w_2$ in the second step of Algorithm~1 in \cite{SB15}. Condition~\eqref{eq_step} for $j=2$ translates now to the two cases of this algorithm, and so on. After~$m$~steps, we obtain the informative confidence bounds for the fallback procedure defined in \cite{SB15}.

\newpage
\section{Simulation results}
\FloatBarrier
\small
\begin{center}
%\begin{table}
\begin{longtable}{c|c|c|c|cc|cc|cc}
&&&& \multicolumn{2}{|c|}{Power} & \multicolumn{2}{|c|}{Mean SCI} & \multicolumn{2}{|c}{\% Finite bounds}   \\
\hline
Scenario & Hypothesis & Ture parameter & $q$ &	ISCI & CSCI &	ISCI &CSCI & ISCI & CSCI \\

\hline
\multirow{8}{*}{1} 		& \multirow{2}{*}{$E_1$}& \multirow{2}{*}{$0$} 								& $10^{-10}$ & 0.812 &  \multirow{2}{*}{0.845} & -0.277  & \multirow{2}{*}{-0.278}  & 100 & \multirow{2}{*}{100}\\
											&  											&  																		& $0.38$ 		 & 0.801 & 												 & -0.288  & 											   & 100 & \\
											& \multirow{2}{*}{$E_2$}& \multirow{2}{*}{$0$} 								& $10^{-10}$ & 0.812 &  \multirow{2}{*}{0.845} & -0.277 & \multirow{2}{*}{-0.278}  & 100 & \multirow{2}{*}{100}\\
											&  											&  																		& $0.38$ 		 & 0.802 &												 & -0.288 & 											   & 100 & \\
											& \multirow{2}{*}{$S_1$}& \multirow{2}{*}{$\delta_e=0.492$} 	& $10^{-10}$ & 0.810 &  \multirow{2}{*}{0.845} & 0.169 & \multirow{2}{*}{0.000}   & 81.2 & \multirow{2}{*}{84.6}\\
											&  											&  																		& $0.38$ 		 & 0.800 &												 & 0.250 &												   & 80.1 & \\
											& \multirow{2}{*}{$S_2$}& \multirow{2}{*}{$\delta_e=0.492$} 	& $10^{-10}$ & 0.810 &  \multirow{2}{*}{0.846} & 0.196 & \multirow{2}{*}{0.000}  & 81.2 & \multirow{2}{*}{84.6}\\
											&  											&  																		& $0.38$ 		 & 0.800 &												 & 0.249 & 											   & 80.2 & \\
							
\hline
\multirow{8}{*}{2} 		& \multirow{2}{*}{$E_1$}& \multirow{2}{*}{$\delta_e=0.492$} 	& $10^{-10}$ & 1.000 &  \multirow{2}{*}{1.000} & 0.210 & \multirow{2}{*}{0.216}& 100 & \multirow{2}{*}{100}\\
											&  											&  																		& $0.38$ 		& 1.000 & 												& 0.151 & 											& 100 & \\
											& \multirow{2}{*}{$E_2$}& \multirow{2}{*}{$\delta_e=0.492$} 	& $10^{-10}$ & 1.000 &  \multirow{2}{*}{1.000} & 0.210 & \multirow{2}{*}{0.216}& 100 & \multirow{2}{*}{100}\\
											&  											&  																		& $0.38$ 		& 1.000 &													& 0.151 & 											& 100 & \\
											& \multirow{2}{*}{$S_1$}& \multirow{2}{*}{$\delta_e=0.492$} 	& $10^{-10}$ & 1.000 &  \multirow{2}{*}{1.000} & 0.207 & \multirow{2}{*}{0.000}& 100 & \multirow{2}{*}{100}\\
											&  											&  																		& $0.38$ 		& 1.000 &													& 0.287 &												& 100 & \\
											& \multirow{2}{*}{$S_2$}& \multirow{2}{*}{$\delta_e=0.492$} 	& $10^{-10}$ & 1.000 &  \multirow{2}{*}{1.000} & 0.207 & \multirow{2}{*}{0.000}& 100 & \multirow{2}{*}{100}\\
											&  											&  																		& $0.38$ 		& 1.000 &													& 0.287 & 											& 100 & \\
							
\hline
\multirow{8}{*}{3} 		& \multirow{2}{*}{$E_1$}& \multirow{2}{*}{$0$}							 	& $10^{-10}$ & 0.855 &  \multirow{2}{*}{0.869} & -0.252 & \multirow{2}{*}{-0.269}& 100 & \multirow{2}{*}{100}\\
											&  											&  																		& $0.38$ 		 & 0.816 & 												 & -0.282 & 											 & 100 & \\
											& \multirow{2}{*}{$E_2$}& \multirow{2}{*}{$\delta_e=0.492$} 	& $10^{-10}$ & 1.000 &  \multirow{2}{*}{1.000} & 0.165 & \multirow{2}{*}{0.151}  & 100 & \multirow{2}{*}{100}\\
											&  											&  																		& $0.38$ 		 & 1.000 &												 & 0.146 & 											   & 100 & \\
											& \multirow{2}{*}{$S_1$}& \multirow{2}{*}{$\delta_e=0.492$} 	& $10^{-10}$ & 0.853 &  \multirow{2}{*}{0.869} & 0.179 & \multirow{2}{*}{0.000}  & 85.5 & \multirow{2}{*}{86.9}\\
											&  											&  																		& $0.38$ 		 & 0.814 &												 & 0.254 &												 & 81.6 & \\
											& \multirow{2}{*}{$S_2$}& \multirow{2}{*}{$\delta_e=0.492$} 	& $10^{-10}$ & 1.000 &  \multirow{2}{*}{1.000} & 0.188 & \multirow{2}{*}{0.000}  & 100 & \multirow{2}{*}{100}\\
											&  											&  																		& $0.38$ 		 & 1.000 &												 & 0.282 & 											   & 100 & \\
							
\hline
\multirow{8}{*}{4} 		& \multirow{2}{*}{$E_1$}& \multirow{2}{*}{$0$}							 	& $10^{-10}$ & 0.797 &  \multirow{2}{*}{0.798} & -0.291 & \multirow{2}{*}{-0.389}& 100 & \multirow{2}{*}{100}\\
											&  											&  																		& $0.38$ 		 & 0.797 & 												 & -0.291 & 											 & 100 & \\
											& \multirow{2}{*}{$E_2$}& \multirow{2}{*}{$0$}							 	& $10^{-10}$ & 0.813 &  \multirow{2}{*}{0.845} & -0.282 & \multirow{2}{*}{-0.386}  & 100 & \multirow{2}{*}{100}\\
											&  											&  																		& $0.38$ 		 & 0.802 &												 & -0.288 & 											   & 100 & \\
											& \multirow{2}{*}{$S_1$}& \multirow{2}{*}{$\delta_e=0.492$} 	& $10^{-10}$ & 0.794 &  \multirow{2}{*}{0.798} & 0.162  & \multirow{2}{*}{0.000}  & 79.7 & \multirow{2}{*}{79.8}\\
											&  											&  																		& $0.38$ 		 & 0.794 &												 & 0.247  &												 & 79.7 & \\
											& \multirow{2}{*}{$S_2$}& \multirow{2}{*}{$0$}							 	& $10^{-10}$ & 0.004 &  \multirow{2}{*}{0.020} & -0.229 & \multirow{2}{*}{-0.174}  & 81.3 & \multirow{2}{*}{84.5}\\
											&  											&  																		& $0.38$ 		 & 0.003 &												 & -0.235 & 											   & 80.2 & \\		
												
\hline
\multirow{8}{*}{5} 		& \multirow{2}{*}{$E_1$}& \multirow{2}{*}{$\delta_e=0.492$} 	& $10^{-10}$ & 1.000 &  \multirow{2}{*}{1.000} & 0.144  & \multirow{2}{*}{-0.376}  & 100 & \multirow{2}{*}{100}\\
											&  											&  																		& $0.38$ 		 & 1.000 & 												 & 0.144  & 											   & 100 & \\
											& \multirow{2}{*}{$E_2$}& \multirow{2}{*}{$0$}							 	& $10^{-10}$ & 0.796 &  \multirow{2}{*}{0.798} & -0.291 & \multirow{2}{*}{-0.391}  & 100 & \multirow{2}{*}{100}\\
											&  											&  																		& $0.38$ 		 & 0.796 &												 & -0.291 & 											   & 100 & \\
											& \multirow{2}{*}{$S_1$}& \multirow{2}{*}{$0$}							 	& $10^{-10}$ & 0.010 &  \multirow{2}{*}{0.013} & -0.203 & \multirow{2}{*}{-0.195}   & 100 & \multirow{2}{*}{100}\\
											&  											&  																		& $0.38$ 		 & 0.010 &												 & -0.203 &												   & 100 & \\
											& \multirow{2}{*}{$S_2$}& \multirow{2}{*}{$0$}							 	& $10^{-10}$ & 0.003 &  \multirow{2}{*}{0.012} & -0.238 & \multirow{2}{*}{-0.195}  & 79.6 & \multirow{2}{*}{79.8}\\
											&  											&  																		& $0.38$ 		 & 0.003 &												 & -0.238 & 											   & 79.6 & \\

\hline
\multirow{8}{*}{6} 		& \multirow{2}{*}{$E_1$}& \multirow{2}{*}{$\delta_e=0.492$} 	& $10^{-10}$ & 1.000 &  \multirow{2}{*}{1.000} & 0.144  & \multirow{2}{*}{-0.376}  & 100 & \multirow{2}{*}{100}\\
											&  											&  																		& $0.38$ 		 & 1.000 & 												 & 0.144  & 											   & 100 & \\
											& \multirow{2}{*}{$E_2$}& \multirow{2}{*}{$\delta_e=0.492$} 	& $10^{-10}$ & 1.000 &  \multirow{2}{*}{1.000} & 0.144 & \multirow{2}{*}{-0.376}  & 100 & \multirow{2}{*}{100}\\
											&  											&  																		& $0.38$ 		 & 1.000 &												 & 0.144 & 											   & 100 & \\
											& \multirow{2}{*}{$S_1$}& \multirow{2}{*}{$0$}							 	& $10^{-10}$ & 0.010 &  \multirow{2}{*}{0.014} & -0.203 & \multirow{2}{*}{-0.195}   & 100 & \multirow{2}{*}{100}\\
											&  											&  																		& $0.38$ 		 & 0.010 &												 & -0.203 &												   & 100 & \\
											& \multirow{2}{*}{$S_2$}& \multirow{2}{*}{$0$}							 	& $10^{-10}$ & 0.010 &  \multirow{2}{*}{0.014} & -0.203 & \multirow{2}{*}{-0.195}  & 100 & \multirow{2}{*}{100}\\
											&  											&  																		& $0.38$ 		 & 0.010 &												 & -0.203 & 											   & 100 & \\

\hline
\multirow{8}{*}{7} 		& \multirow{2}{*}{$E_1$}& \multirow{2}{*}{$\delta_e=0.492$} 	& $10^{-10}$ & 1.000 &  \multirow{2}{*}{1.000} & 0.145  & \multirow{2}{*}{-0.363}  & 100 & \multirow{2}{*}{100}\\
											&  											&  																		& $0.38$ 		 & 1.000 & 												 & 0.144  & 											   & 100 & \\
											& \multirow{2}{*}{$E_2$}& \multirow{2}{*}{$\delta_e=0.492$} 	& $10^{-10}$ & 1.000 &  \multirow{2}{*}{1.000} & 0.162 & \multirow{2}{*}{-0.363}  & 100 & \multirow{2}{*}{100}\\
											&  											&  																		& $0.38$ 		 & 1.000 &												 & 0.151 & 											   & 100 & \\
											& \multirow{2}{*}{$S_1$}& \multirow{2}{*}{$\delta_e=0.492$} 	& $10^{-10}$ & 1.000 &  \multirow{2}{*}{1.000} & 0.180 & \multirow{2}{*}{0.000}   & 100 & \multirow{2}{*}{100}\\
											&  											&  																		& $0.38$ 		 & 1.000 &												 & 0.280 &												   & 100 & \\
											& \multirow{2}{*}{$S_2$}& \multirow{2}{*}{$0$}							 	& $10^{-10}$ & 0.018 &  \multirow{2}{*}{0.025} & -0.182 & \multirow{2}{*}{-0.171}  & 100 & \multirow{2}{*}{100}\\
											&  											&  																		& $0.38$ 		 & 0.012 &												 & -0.196 & 											   & 100 & \\				
\caption{Performance of the informative SCIs (ISCI) and compatible SCIs (CSCI) under different setups.}
\label{tab_simres_scenarios}																						
\end{longtable}
%\end{table}
\end{center}
\end{document}